\documentclass[leqno]{article} %

\usepackage[utf8]{inputenc}
\usepackage{multirow}
\usepackage[left=3.2cm,right=3.2cm,top=3.2cm,bottom=4.2cm]{geometry}

\usepackage[utf8]{inputenc}
\usepackage{color,xspace}
\usepackage{amsfonts}
\usepackage{amssymb}

\usepackage{amsmath}
\usepackage{amsfonts}
\usepackage{multicol}
\usepackage{tikz}
\usepackage{prettyref}
\usepackage{pgfplots}
\usepackage{microtype}
\usepackage{enumerate}
\usepackage{listings}

\usepackage{pgfplots}
\usepackage{pgfplotstable}
\usepackage{pgf}
\usepackage{float}

\usepackage[amsmath,hyperref,thmmarks,amsthm]{ntheorem}

\theoremseparator{:}
\theoremstyle{plain}
\theorembodyfont{\slshape}
\newtheorem{theorem}{Theorem}[section]

\newtheorem{lemma}[theorem]{Lemma}

\theoremstyle{definition}

\definecolor{blue}{rgb}{0.211,0.211,0.856}
\definecolor{red}{rgb}{0.856,0.2,0.2}
\newcommand{\blue}{\color{blue}}
\newcommand{\red}{\color{red}}

\definecolor{darkgreen}{rgb}{0.1,0.656,0.1}
\newcommand{\darkgreen}{\color{darkgreen}}

\definecolor{green}{rgb}{0.1,0.656,0.1}

\definecolor{lgreen}{rgb}{0.411,1.0,0.411}
\definecolor{lred}{rgb}{1.0,0.711,0.501}

\definecolor{lgrey}{rgb}{0.5,0.5,0.5}
\definecolor{llgrey}{rgb}{0.68,0.68,0.68}

\newenvironment{SE}{\noindent\color{red} SE : }{}

\newenvironment{aw}{\noindent\color{magenta} AW :  }{}

\newcommand\ie{i.\,e.,\xspace}

\newcommand\eg{e.\,g.\xspace}
\newrefformat{thm}{Theorem~\ref{#1}}
\newrefformat{lem}{Lemma~\ref{#1}}
\newrefformat{dfn}{Definition~\ref{#1}}
\newrefformat{cor}{Corollary~\ref{#1}} 
\newrefformat{prop}{Proposition~\ref{#1}}
\newrefformat{sec}{Section~\ref{#1}}
\newrefformat{kap}{Chapter~\ref{#1}}
\newrefformat{bsp}{Example~\ref{#1}}
\newrefformat{rem}{Remark~\ref{#1}}
\newrefformat{fig}{Figure~\ref{#1}}
\newrefformat{eq}{(\ref{#1})}
\newrefformat{ex}{Example~\ref{#1}}
\newrefformat{tab}{Table~\ref{#1}}
\newcommand{\QuickMergesort}{{QuickMergesort}\xspace}

\newcommand{\QuickXsort}{{QuickXsort}\xspace}

\newcommand{\bMQMS}{\ensuremath{\mathrm{bMQMS}\xspace}}
\newcommand{\MQMS}{\ensuremath{\mathrm{MQMS}\xspace}}
\newcommand{\uMQMS}[1]{\ensuremath{\mathrm{MQMS}_{#1}\xspace}}
\newcommand{\MS}{\ensuremath{\mathrm{MS}\xspace}}
\newcommand{\MOM}{\ensuremath{\mathrm{MoM}\xspace}}
\newcommand{\N}{\mathbb{N}}
\newcommand{\Z}{\mathbb{Z}}

\newcommand{\R}{\mathbb{R}}

\newcommand{\Oh}{\mathcal{O}}

\newcommand{\floor}[1]{\left\lfloor\mathinner{#1} \right\rfloor}
\newcommand{\ceil}[1]{\left\lceil\mathinner{#1} \right\rceil}

\newcommand{\oneset}[1]{\left\{\, \mathinner{#1} \,\right\}}

\newcommand{\dd}{:$\!$:}
\newcommand{\stdsort}{\texttt{std:$\!$:sort}\xspace}
\newcommand{\stdpartialsort}{\texttt{std:$\!$:partial\_sort}\xspace}
\newcommand{\stdnth}{\texttt{std:$\!$:nth\_element}\xspace}
\newcommand{\stdstablesort}{\texttt{std:$\!$:stable\_sort}\xspace}

 \allowdisplaybreaks
 \usepackage[hidelinks]{hyperref}

\begin{document}
 \allowdisplaybreaks
   \title{Worst-Case Efficient Sorting with QuickMergesort}

 \author{Stefan Edelkamp\thanks{King's College London, UK.} \and  Armin Wei\ss\thanks{Universit{\"a}t Stuttgart, Germany. Supported by the DFG grant DI 435/7-1.}}

\date{}

\maketitle

\definecolor{mygreen}{rgb}{0,0.6,0}
\definecolor{mygray}{rgb}{0.5,0.5,0.5}
\definecolor{mymauve}{rgb}{0.58,0,0.82}

\begin{abstract}
The two most prominent solutions for the sorting problem are Quicksort and Mergesort. While Quicksort is very fast on average, Mergesort additionally gives worst-case guarantees, but needs extra space for a linear number of elements. Worst-case efficient in-place sorting, however, remains a challenge: the standard solution, Heapsort, suffers from a bad cache behavior and is also not overly fast for in-cache instances.
	
In this work we present median-of-medians QuickMergesort (MoMQuickMergesort), a new variant of QuickMergesort, which combines Quicksort with Mergesort allowing the latter to be implemented in place. Our new variant applies the median-of-medians algorithm for selecting pivots in order to circumvent the quadratic worst case. Indeed, we show that it uses at most $n \log n + 1.6n$ comparisons for $n$ large enough. 
	
We experimentally confirm the theoretical estimates and show that the new algorithm outperforms Heapsort by far and is only around 10\% slower than Introsort (\stdsort implementation of stdlibc++), which has a rather poor guarantee for the worst case. We also simulate the worst case, which is only around 10\% slower than the average case. In particular, the new algorithm is a natural candidate to replace Heapsort as a worst-case stopper in Introsort.

\smallskip
\noindent \textbf{keywords:}
in-place sorting, quicksort, mergesort, analysis of algorithms
\end{abstract}

\section{Introduction}

Sorting elements of some totally ordered universe always has been among the most important tasks carried out on computers. Comparison based sorting of $n$ elements requires at least $\log n! \approx n\log n - 1.44n$ comparisons (where $\log$ is base 2). Up to constant factors this bound is achieved by the classical sorting algorithms Heapsort, Mergesort, and Quicksort. While Quicksort usually is considered the fastest one, the $\Oh(n\log n)$-bound applies only for its average case (both for the number of comparisons and running time)~-- in the worst-case it deteriorates to a $\Theta(n^2)$ algorithm. 
The standard approach to prevent such a worst-case is Musser's Introsort \cite{Mus97}: whenever the recursion depth of Quicksort becomes too large, the algorithm switches to Heapsort (we call this the \emph{worst-case stopper}). This works well in practice for most instances. However, on small instances Heapsort is already considerably slower than Quicksort (in our experiments more than 30\% for $n=2^{10}$) and on larger instances it suffers from its poor cache behavior (in our experiments more than eight times slower than Quicksort for sorting $2^{28}$ elements). This is also the reason why in practice it is mainly used as a worst-case stopper in Introsort.

Another approach for preventing Quicksort's worst case is by using the median-of-medians algorithm \cite{BFPRT73} for pivot selection.
 However, choosing the pivot as median of the whole array yields a bad average (and worst-case) running time. On the other hand, when choosing the median of a smaller sample as pivot, the average performance becomes quite good \cite{Kurosawa16}, but the guarantees for the worst case become even worse.

The third algorithm, Mergesort, is almost optimal in terms of comparisons: it uses only $n \log n- 0.91n$ comparisons in the worst-case to sort $n$ elements. Moreover, it performs well in terms of running time. Nevertheless, it is not used as worst-case stopper for Introsort because it needs extra space for a linear number of data elements. In recent years, several in-place (we use the term for at most logarithmic extra space) variants of Mergesort appeared, both stable ones (meaning that the relative order of elements comparing equal is not changed) \cite{HuangL92,KimK08,GeffertKP00} and unstable ones \cite{Chen06,ElmasryKS12,GeffertKP00,KatajainenPT96}. Two of the most efficient implementations of stable variants are Wikisort \cite{wikisort} (based on \cite{KimK08}) and Grailsort \cite{grailsort} (based on \cite{HuangL92}). An example for an unstable in-place Mergesort implementation is in-situ Mergesort \cite{ElmasryKS12}. It uses Quick/Introselect \cite{Hoare61_find} (\stdnth) to find the median of the array. Then it partitions the array according to the median (\ie move all smaller elements to the right and all greater elements to the left). Next, it sorts one half with Mergesort using the other half as temporary space, and, finally, sort the other half recursively. Since the elements in the temporary space get mixed up (they are used as ``dummy'' elements), this algorithm is not stable.
In-situ Mergesort gives an $\Oh(n \log n)$ bound for the worst case. As validated in our experiments all the in-place variants are considerably slower than ordinary Mergesort.

When instead of the median an arbitrary element is chosen as the pivot, we obtain QuickMergesort \cite{EdelkampW14}, which is faster on average~-- with the price that the worst-case can be quadratic. 
QuickMergesort follows the more general concept of QuickXsort\cite{EdelkampW14}: first, choose a pivot element and partition the array according to it. Then, sort one part with X and, finally, the other part recursively with {QuickXsort}.
As for QuickMergesort, the part which is currently not being sorted can be used as temporary space for X.

Other examples for QuickXsort are QuickHeapsort \cite{CantoneC02,DiekertW13Quick} and QuickWeakheapsort \cite{ES02,EdelkampW14} and Ultimate Heapsort \cite{Katajainen98}. 
QuickXsort with median-of-$\sqrt{n}$ pivot selection uses at most $n \log n + cn + o(n)$ comparisons on average to sort $n$ elements given that X also uses at most $n \log n + cn
+o(n)$ comparisons on average~\cite{EdelkampW14}. 
Moreover, recently Wild \cite{Wild18} showed that, if the pivot is selected as median of some constant size sample, then the average number of comparisons of QuickXsort is only some small linear term (depending on the sample size) above the average number of comparisons of $X$ (for the median-of-three case see also \cite{EdelkampW18QMSArxiv}).
However, as long as no linear size samples are used for pivot selection, QuickXsort does not provide good bounds for the worst case. This defect is overcome in Ultimate Heapsort \cite{Katajainen98} by using the median of the whole array as pivot. In Ultimate Heapsort the median-of-medians algorithms \cite{BFPRT73} (which is linear in the worst case) is used for finding the median, leading to an $n \log n + \Oh(n)$ bound for the number of comparisons. Unfortunately, due to the large constant of the median-of-medians algorithm, the $\Oh(n)$-term is quite big. 

\smallskip
\paragraph{Contribution.}
In this work we introduce median-of-medians QuickMergesort (MoMQuickMergesort) as a variant of QuickMergesort using the median-of-medians algorithms for pivot selection. The crucial observation is that it is not necessary to use the median of the whole array as pivot, but only the guarantee that the pivot is not very far off the median. This observation allows to apply the median-of-medians algorithm to smaller samples leading to both a better average- and worst-case performance. Our algorithm is based on a merging procedure introduced by Reinhardt \cite{Reinhardt92}, which requires less temporary space than the usual merging.
A further improvement, which we call \emph{undersampling} (taking less elements for pivot selection into account), allows to reduce the worst-case number of comparisons down to $n \log n + 1.59n + \Oh(n^{0.8})$. Moreover, we heuristically estimate the average case as $n\log n + 0.275 n + o(n)$ comparisons. The good average case comes partially from the fact that we introduce a new way of adaptive pivot selection for the median-of-medians algorithm (compare to \cite{Alexandrescu17}).
Our experiments confirm the theoretical and heuristic estimates and also show that MoMQuickMergesort is competitive to other algorithms (for $n=2^{28}$ more than 7 times faster than Heapsort and around 10\% slower than Introsort (\stdsort~-- throughout this refers to its libstdc++ implementation)). 
Moreover, we apply MoMQuickMergesort (instead of Heapsort) as a worst-case stopper for Introsort (\stdsort). The results are striking: on special permutations, the new variant is up to six times faster than the original version of \stdsort. 

\smallskip
\paragraph{Outline.}
In \prettyref{sec:prelims}, we recall QuickMergesort and the median-of-medians algorithm. In \prettyref{sec:momqms}, we describe median-of-medians QuickMergesort, introduce the improvements and analyze the worst-case and average-case behavior. Finally, in \prettyref{sec:experiments}, we present our experimental results.

\section{Preliminaries}\label{sec:prelims}
Throughout we use standard $\Oh$ and $\Theta$ notation as defined \eg\ in \cite{CLRS09}. The logarithm $\log$ always refers to base 2. For a background on Quicksort and Mergesort we refer to \cite{CLRS09} or \cite{Knu73}. A \emph{pseudomedian} of nine (resp.\ fifteen) elements is computed as follows: group the elements in groups of three elements and compute the median of each group. The \emph{pseudomedian} is the median of these three (resp.\ five) medians. 

Throughout, in our estimates we assume that the median of three (resp.\ five) elements is computed using three (resp.\ seven) comparisons no matter on the outcome of previous comparisons. This allows a branch-free implementation of the comparisons.

In this paper we have to deal with simple recurrences of two types, which both have straightforward solutions:
\begin{lemma}\label{lem:recurrence}
Let $0< \alpha,\beta, \delta$ with $\alpha + \beta < 1$, $\gamma =1-\alpha$, and $A,C,D,N_0 \in \N$ and
	\begin{align*}
	T(n) &\leq T(\ceil{\alpha n} + A)  + T(\ceil{\beta n} + A) + Cn + D\\
	Q(n) &\leq Q(\ceil{\alpha n} + A)  + \gamma n \log(\gamma n)  + Cn + \Oh(n^\delta)
	\end{align*}
	for $n \geq N_0$ and $T(n), Q(n) \leq D$ for $n \leq N_0$ ($N_0$ large enough). Moreover, let $\zeta \in \R$  such that $\alpha^\zeta + \beta^\zeta = 1$ (notice that $\zeta < 1$). Then
	\begin{align*}
	T(n) &\leq   \frac{Cn}{1-\alpha - \beta} + \Oh(n^\zeta) \qquad\qquad \text{and}\\
	Q(n) &\leq n \log n \!+\!  \left(\frac{ \alpha  \log \alpha }{\gamma}\!+\! \log \gamma  \!+\!\frac{C}{\gamma}\right) n\!+ \!  \Oh(n^\delta). 
	\end{align*}	
\end{lemma}
\begin{proof}	
	It is well-known that $T(n)$ has a linear solution. Therefore, (after replacing $T(n)$ by a reasonably smooth function) $T(\ceil{\alpha n} + A)$ and $T(\floor{\alpha n})$ differ by at most some constant. Thus, after increasing $D$, we may assume that $T(n)$ is of the simpler form
	\begin{align}
	T(n) &\leq T(\alpha n)  + T(\beta n) + Cn + D.\label{eq:rec3}
	\end{align}
	We can split \prettyref{eq:rec3} into two recurrences 
	\begin{align*}
	T_C(n) &\leq T(\alpha n)  + T(\beta n)  + Cn \qquad \qquad \text{and} &
	T_D(n) &\leq T(\alpha n)  + T(\beta n)  + D
	\end{align*}
	with $T_C(n) = 0$ and  $T_D(n) \leq D$ for $n \leq N_0$. For $T_C$ we get the solution $T_C(n) \leq \frac{Cn}{1-\alpha-\beta}$.  By the generalized Master theorem \cite{Kao97}, it follows that $T_D \in \Oh(n^\zeta)$ where $\zeta \in \R$ satisfies $\alpha^\zeta + \beta^\zeta = 1$. Thus, 
	\[
	T(n) \leq  \frac{Cn}{1-\alpha - \beta} + \Oh(n^\zeta).
	\]
	
	Now, let us consider the recurrence for $Q(n)$. With the same argument as before we have $Q(n) \leq  Q(\alpha n)  + \gamma n \log(\gamma n)  + Cn +   \Oh(n^\delta)$. Thus, we obtain
	\begin{align*}
	Q(n) 	&\leq \sum_{i=0}^{\log_\alpha n} \Bigl( \alpha^i \gamma n \log (\alpha^i\gamma n) + C \alpha^in +   \Oh((\alpha^in)^\delta)\Bigr)\\
	&= n \sum_{i\geq 0} \Bigl( \alpha^i \gamma \left( \log n +  i \log (\alpha) + \log \gamma \right) + C \alpha^i\Bigr) +   \Oh(n^\delta)\\
	&= \frac{\gamma }{1-\alpha} n \log n +  \left(\frac{ \alpha \gamma \log \alpha }{(\alpha - 1)^2}+ \frac{ \gamma \log \gamma }{1-\alpha}  +\frac{C}{1-\alpha}\right) n+     \Oh(n^\delta)\\
	&= n \log n +  \left(\frac{ \alpha  \log \alpha }{1-\alpha}+ \log \gamma  +\frac{C}{1-\alpha}\right) n +   \Oh(n^\delta).
	\end{align*}
	This proves \prettyref{lem:recurrence}.
\end{proof}

\subsection{QuickMergesort}\label{sec:quickXsort}\label{sec:qms}

QuickMergesort follows the design pattern of QuickXsort: let X be some sorting algorithm (in our case X = Mergesort).  {QuickXsort} works as follows:
first, choose some pivot element and partition
the array according to this pivot, i.\,e., rearrange it
such that all elements left of the pivot are less or equal and all
elements on the right are greater than or equal to the pivot element. Then, choose one part of the array and sort
it with the algorithm X. After that, sort the other part of the array
recursively with {QuickXsort}. 
The main advantage of this procedure is that the part of the array
that is not being sorted currently can be used as temporary memory for
the algorithm X. This yields fast in-place variants for various
\emph{external} sorting algorithms such as {Mergesort}. The idea is
that whenever a data element should be moved to the extra (additional
or external) element space, instead it is swapped with the data
element occupying the respective position in part of the array which
is used as temporary memory.

The most promising example for \QuickXsort{} is {QuickMergesort}. For the {Mergesort}
part we use standard (top-down) {Mergesort}, which can be implemented
using $m$ extra element spaces to merge two arrays of length $m$:
after the partitioning, one part of the array -- for a simpler description we assume the first
part -- has to be sorted with {Mergesort} (note, however, that any of the two sides can be sorted with Mergesort as long as the other side contains at least $n/3$ elements). 
In order to do so, the
second half of this first part is sorted recursively with {Mergesort}
while moving the elements to the back of the whole array. The elements
from the back of the array are inserted as dummy elements into the
first part. Then, the first half of the first part is sorted recursively
with {Mergesort} while being moved to the position of the former
second half of the first part. Now, at the front of the array, there is enough space
(filled with dummy elements) such that the two halves can be merged.
The executed stages of the algorithm {QuickMergesort} are illustrated in \prettyref{fig:sample}.
\begin{figure}[t]
	\footnotesize
\begin{center} \hspace{0.12cm}
\begin{tikzpicture}[scale=.75]
{
\draw(-2,4) -- (6.4,4);
\draw(-2,4.7) -- (6.4,4.7);
\draw(-2,4) -- (-2,4.7);
\draw(-1.3,4) -- (-1.3,4.7);
\node[] (2) at (-0.95,4.35) {$11$};
\draw(-0.6,4) -- (-0.6,4.7);
\node[] (3) at (-0.25,4.35) {$4$};
\draw(0.1,4) -- (0.1,4.7);
\node[] (4) at (.45,4.35) {$5$};
\draw(.8,4) -- (.8,4.7);
\node[] (5) at (1.15,4.35) {$6$};
\draw(1.5,4) -- (1.5,4.7);
\node[] (6) at (1.85,4.35) {$10$};
\draw(2.2,4) -- (2.2,4.7);
\node[] (7) at (2.55,4.35) {$9$};
\draw(2.9,4) -- (2.9,4.7);
\node[] (8) at (3.25,4.35) {$2$};
\draw(3.6,4) -- (3.6,4.7);
\node[] (9) at (3.95,4.35) {$3$};
\draw(4.3,4) -- (4.3,4.7);
\node[] (10) at (4.65,4.35) {$1$};
\draw(5.0,4) -- (5.0,4.7);
\node[] (11) at (5.35,4.35) {$0$};
\draw(5.7,4) -- (5.7,4.7);
\node[] (12) at (6.05,4.35) {$8$};
\draw(6.4,4) -- (6.4,4.7);
\node[] (1) at (-1.65,4.35) {${\red 7}$};
}
\end{tikzpicture}

partitioning leads to 

\vspace{0.1cm}
\begin{tikzpicture}[scale=.75]
{
\draw(-2,4) -- (6.4,4);
\draw(-2,4.7) -- (6.4,4.7);
\draw(-2,4) -- (-2,4.7);
\node[] (1) at (-1.65,4.35) {$3$};
\draw(-1.3,4) -- (-1.3,4.7);
\node[] (2) at (-0.95,4.35) {$2$};
\draw(-0.6,4) -- (-0.6,4.7);
\node[] (3) at (-0.25,4.35) {$4$};
\draw(0.1,4) -- (0.1,4.7);
\node[] (4) at (.45,4.35)   {$5$};
\draw(.8,4) -- (.8,4.7);
\node[] (5) at (1.15,4.35)  {$6$};
\draw(1.5,4) -- (1.5,4.7);
\node[] (6) at (1.85,4.35)  {$0$};
\draw(2.2,4) -- (2.2,4.7);
\node[] (7) at (2.55,4.35)  {$1$};
\draw(2.9,4) -- (2.9,4.7);
\draw[thick](3.6,4) -- (3.6,4.7);
\node[] (9) at (3.95,4.35)  {$9$};
\draw(4.3,4) -- (4.3,4.7);
\node[] (10) at (4.65,4.35) {$10$};
\draw(5.0,4) -- (5.0,4.7);
\node[] (11) at (5.35,4.35) {$11$};
\draw(5.7,4) -- (5.7,4.7);
\node[] (12) at (6.05,4.35)  {$8$};
\draw(6.4,4) -- (6.4,4.7);
\node[] (8) at (3.25,4.35) {${\red 7}$};
}

{
\node[]  at (1.5,3.85) {$\underbrace{\hspace{21.2mm}}$};
\node[] (ubxx) at (1.5,3.85) {};
\node[]  at (0.2,3.5) {\footnotesize sort recursively};
\node[]  at (3.1,3.5) {\footnotesize with Mergesort};

\draw[thick] (.8,4) -- (.8,4.7);
\draw[thick] (2.9,4) -- (2.9,4.7);
}

\begin{scope}[shift={(0,-1.7)}]
{
\draw(-2,4) -- (6.4,4);
\draw(-2,4.7) -- (6.4,4.7);
\draw(-2,4) -- (-2,4.7);
\node[] (1) at (-1.65,4.35) {$3$};
\draw(-1.3,4) -- (-1.3,4.7);
\node[] (2) at (-0.95,4.35) {$2$};
\draw(-0.6,4) -- (-0.6,4.7);
\node[] (3) at (-0.25,4.35) {$4$};
\draw(0.1,4) -- (0.1,4.7);
\node[] (4) at (.45,4.35) {$11$};
\draw(.8,4) -- (.8,4.7);
\node[] (5) at (1.15,4.35) {$9$};
\draw(1.5,4) -- (1.5,4.7);
\node[] (6) at (1.85,4.35) {$10$};
\draw(2.2,4) -- (2.2,4.7);
\node[] (7) at (2.55,4.35) {$8$};
\draw(2.9,4) -- (2.9,4.7);
\node[] (1) at (3.25,4.35) {\red $7$};
\draw[thick](3.6,4) -- (3.6,4.7);
\node[] (9) at (3.95,4.35)  {\blue $0$};
\draw(4.3,4) -- (4.3,4.7);
\node[] (10) at (4.65,4.35) {\blue $1$};
\draw(5.0,4) -- (5.0,4.7);
\node[] (11) at (5.35,4.35) {\blue $5$};
\draw(5.7,4) -- (5.7,4.7);
\node[] (12) at (6.05,4.35) {\blue $6$};
\draw(6.4,4) -- (6.4,4.7);
}
{  
\draw[->] (ubxx) ..controls +(0,-0.9) and (5,5.1).. (5.0,4.7);}
{
  \node[] (ubxxx) at (-0.95,3.85) {};
\node[] (1) at (2,3.48) {\footnotesize sort recursively with Mergesort};
\draw[thick] (.8,4) -- (.8,4.7);
\draw[thick] (2.9,4) -- (2.9,4.7);
\node[] (1) at (-0.95,3.85) {$\underbrace{\hspace{15.9mm}}$};}
\end{scope}

\begin{scope}[shift={(0,-3.4)}]

{
\draw(-2,4) -- (6.4,4);
\draw(-2,4.7) -- (6.4,4.7);
\draw(-2,4) -- (-2,4.7);
\node[] (1) at (-1.65,4.35) { $9$};
\draw(-1.3,4) -- (-1.3,4.7);
\node[] (2) at (-0.95,4.35) { $10$};
\draw(-0.6,4) -- (-0.6,4.7);
\node[] (3) at (-0.25,4.35) { $8$};
\draw(0.1,4) -- (0.1,4.7);
\node[] (4) at (.45,4.35) {$11$};
\draw(.8,4) -- (.8,4.7);
\node[] (5) at (1.15,4.35) {\darkgreen$2$};
\draw(1.5,4) -- (1.5,4.7);
\node[] (6) at (1.85,4.35) {\darkgreen$3$};
\draw(2.2,4) -- (2.2,4.7);
\node[] (7) at (2.55,4.35) {\darkgreen $4$};
\draw(2.9,4) -- (2.9,4.7);
\node[] (1) at (3.25,4.35) {\red $7$};
\draw[thick](3.6,4) -- (3.6,4.7);
\node[] (9) at (3.95,4.35)  {\blue $0$};
\draw(4.3,4) -- (4.3,4.7);
\node[] (10) at (4.65,4.35) {\blue $1$};
\draw(5.0,4) -- (5.0,4.7);
\node[] (11) at (5.35,4.35) {\blue $5$};
\draw(5.7,4) -- (5.7,4.7);
\node[] (12) at (6.05,4.35) {\blue $6$};
\draw(6.4,4) -- (6.4,4.7);
}

{  
\draw[->] (ubxxx) ..controls +(0,-0.9) and (1.85,5.1).. (1.85,4.7);}

{\node[] (1) at (1.85,3.85) {$\underbrace{\hspace{15.9mm}}$};
\node[] (1) at (5.0,3.85) {$\underbrace{\hspace{21.2mm}}$};
  \node[] (ubxxxx) at (1.85,3.85) {};
    \node[] (ubxxxxx) at (5,3.85) {};
\node[] (uba) at (-1,3.38) {\footnotesize merge two parts};
}

\end{scope}

\begin{scope}[shift={(0,-5.1)}]

{
\draw[-] (ubxxxxx) ..controls +(0,-0.5) and (.45,5.4) .. (.45,4.9);
\draw[-] (ubxxxx) ..controls +(0,-0.5) and (.45,5.4) .. (.45,4.9);
\draw[->] (.45,5)-- (.45,4.7);
}

{
\draw(-2,4) -- (6.4,4);
\draw(-2,4.7) -- (6.4,4.7);
\draw(-2,4) -- (-2,4.7);

\draw(-1.3,4) -- (-1.3,4.7);

\draw(-0.6,4) -- (-0.6,4.7);

\draw(0.1,4) -- (0.1,4.7);

\draw(.8,4) -- (.8,4.7);

\draw(.8,4) -- (.8,4.7);

\draw(1.5,4) -- (1.5,4.7);

\draw(2.2,4) -- (2.2,4.7);

\draw(2.9,4) -- (2.9,4.7);

\draw[thick](3.6,4) -- (3.6,4.7);

\draw(4.3,4) -- (4.3,4.7);

\draw(5.0,4) -- (5.0,4.7);

\draw(5.7,4) -- (5.7,4.7);

\draw(6.4,4) -- (6.4,4.7);
}
{

  \node[] (6) at (-1.65,4.35) {\blue $0$};
    \node[] (7) at (-0.95,4.35) {\blue  $1$};
      \node[] (1) at (-0.25,4.35) {\darkgreen $2$};
      \node[] (4) at (.45,4.35) {\darkgreen $3$};
      \node[] (5) at (1.15,4.35) {\darkgreen $4$};
      \node[] (1) at (1.85,4.35) {\blue $5$};
        \node[] (9) at (2.55,4.35) {\blue $6$};
\node[] (10) at (3.25,4.35) {\red $7$};
  \node[] (11) at (3.95,4.35) {$11$};
   \node[] (12) at (4.65,4.35) {$9$};
  \node[] (2) at (5.35,4.35) {$8$};
  \node[] (3) at (6.05,4.35) {$10$};

}

{\node[] (1) at (5.0,3.85) {$\underbrace{\hspace{21.2mm}}$};
\node[] (1) at (3.25,3.5) {\footnotesize sort recursively with \QuickMergesort};}

\end{scope}
\end{tikzpicture}%
\end{center}
\vspace{-4mm}
\caption{\small Example for the execution of {QuickMergesort}. Here 7 is chosen as pivot.}
\label{fig:sample}
\end{figure}
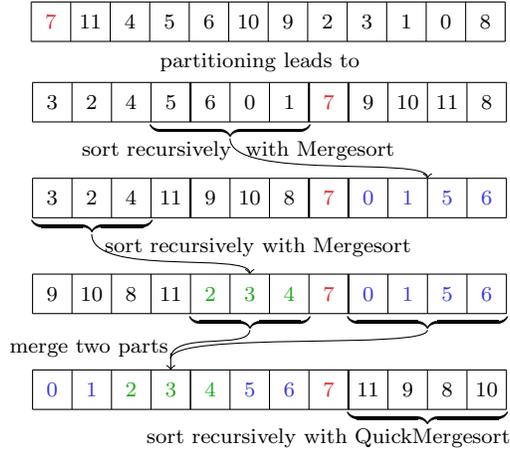

\subsection{The median-of-medians algorithm}\label{sec:mom}
The median-of-medians algorithm solves the selection problem: given an array $A[1,\dots, n]$ and an integer $k\in \oneset{1, \dots, n}$ find the $k$-th element in the sorted order of $A$.
For simplicity let us assume that all elements are distinct~-- in \prettyref{sec:duplicates} we show how to deal with the general case with duplicates.

The basic variant of the median-of-medians algorithm \cite{BFPRT73} (see also \cite[Sec.\ 9.3]{CLRS09}) works as follows: first, the array is grouped into blocks of five elements. From each of these blocks the median is selected and then the median of all these medians is computed recursively. This yields a provably good pivot for performing a partitioning step. Now, depending on which side the $k$-th element is, recursion takes place on the left or right side. It is well-known that this algorithm runs in linear time with a rather big constant in the $\Oh$-notation. We use a slight improvement:

\smallskip
\paragraph{Repeated step algorithm.}
Instead of grouping into blocks of $5$ elements, we follow \cite{ChenD15} and group into blocks of $9$ elements and take the pseudomedian (``ninther'') into the sample for pivot selection. This method guarantees that every element in the sample has $4$ elements less or equal and $4$ element greater or equal to it. Thus, when selecting the pivot as median of the sample of $n/9$ elements, the guarantee is that at least $2n/9$ elements are less or equal and the same number greater or equal to the pivot. Since there might remain 8 elements outside the sample we obtain the recurrence
$$T_{\MOM}\!\!\;(n) \leq T_{\MOM}\!\!\left(\floor{\!\frac{7n}{9}}\!+\!8\!\right) + T_{\MOM}\!\left(\floor{\frac{n}{9}}\right) +  \frac{20n}{9},$$
 where ${4n}/{3}$ of ${20n}/{9}$ is due to finding the pseudomedians and ${8n}/{9}$ is for partitioning the remaining (non-pseudomedian) elements according to the pivot (notice that also some of the other elements are already known to be greater/smaller than the pivot; however, using this information would introduce a huge bookkeeping overhead). Thus, by \prettyref{lem:recurrence}, we have:
\begin{lemma}[\cite{Alexandrescu17,ChenD15}]\label{lem:mom}
	$T_{\MOM}(n) \leq 20n + \Oh(n^\zeta) $ where $\zeta \approx 0.78$  satisfies $(7/9)^\zeta + (1/9)^\zeta = 1$.
\end{lemma}

\smallskip
\paragraph{Adaptive pivot selection.} For our implementation we apply a slight improvement over the basic median-of-medians algorithm by using the approach of adaptive pivot selection, which is first used in the Floyd-Rivest algorithm \cite{FloydR75a,FloydR75}, later applied to smaller samples for Quickselect \cite{MartinezPV04,MartinezPV10}, and recently applied to the median-of-medians algorithm \cite{Alexandrescu17}. 
However, we use a different approach than in \cite{Alexandrescu17}: in any case we choose the sample of size $n/9$ as pseudomedians of nine. Now, if the position we are looking for is on the far left (left of position $2n/9$), we do not choose the median of the sample as pivot but a smaller position:
for searching the $k$-th element with $k \leq 2n/9$, we take the $\ceil{k/4}$-th element of the sample as pivot. Notice that for $k = 2n/9$, this is exactly the median of the sample. Since every element of the sample carries at least four smaller elements with it, this guarantees that 
$\frac{k}{4} \cdot 4 = k$ elements are smaller than or equal to the pivot~-- so the $k$-th element will lie in the left part after partitioning (which is presumably the smaller one). 
Likewise when searching a far right position, we proceed symmetrically.

Notice that this optimization does not improve the worst-case but the average case (see \prettyref{sec:average}).

\subsection{Dealing with duplicate elements}\label{sec:duplicates}

With duplicates we mean that not all elements of the input array are distinct. The number of comparisons for finding the median of three (resp.\ five) elements does not change in the presence of duplicates. However, duplicates can lead to an uneven partition. 
The standard approach in Quicksort and Quickselect for dealing with duplicates is due to Bentley and McIlroy \cite{BentleyM93}: in each partitioning step the elements equal to the pivot are placed in a third partition in the middle of the array. Recently, another approach appeared in the Quicksort implementation pdqsort \cite{pdqsort}.
Instead of three-way partitioning it applies the usual two-way partitioning moving elements equal to the pivot always to the right side. This method is also applied recursively~-- with one exception: if the new pivot is equal to an old pivot (this can be tested with one additional comparison), then all elements equal to the pivot are moved to the left side, which then can be excluded from recursion. 

We propose to follow the latter approach: usually all elements equal to the pivot are moved to the right side -- possibly leading to an even unbalanced partitioning.
However, whenever a partitioning step is very uneven (outside the guaranteed bounds for the pivot in the median-of-medians algorithm), we know that this must be due to many duplicate elements. In this case we immediately partition again with the same pivot but moving equal elements to the left.

\section{Median-of-Medians QuickMergesort}\label{sec:worstcase}\label{sec:momqms}

Although QuickMergesort has an $\Oh(n^2)$ worst-case running time, it is quite simple to guarantee a worst-case number of comparisons of $n \log n + \Oh(n)$: just choose the median of the whole array as pivot. This is essentially how in-situ Mergesort~\cite{ElmasryKS12} works. The most efficient way for finding the median is using Quickselect \cite{Hoare61_find} as applied in in-situ Mergesort. However, this does not allow the desired bound on the number of comparisons (even not when using Introselect as in~\cite{ElmasryKS12}). Alternatively, we can use the median-of-medians algorithm described in \prettyref{sec:mom}, which, while having a linear worst-case running time, on average is quite slow. In this section we describe a variation of the median-of-medians approach which combines an $n \log n + \Oh(n)$ worst-case number of comparisons with a good average performance (both in terms of running time and number of comparisons).

\subsection{Basic version}\label{sec:basicQMS}

The crucial observation is that it is not necessary to use the actual median as pivot (see also our preprint \cite{EdelkampW18QMSArxiv}). As remarked in \prettyref{sec:quickXsort}, the larger of the two sides of the partitioned array can be sorted with Mergesort as long as the smaller side contains at least one third of the total number of elements. Therefore, it suffices to find a pivot which guarantees such a partition. For doing so, we can apply the idea of the median-of-medians algorithm: for sorting an array of $n$ elements, we choose first $n/3$ elements as median of three elements each. Then, the median-of-medians algorithm is used to find the median of those $n/3$ elements. This median becomes the next pivot. Like for the median-of-medians algorithm, this ensures that at least $2\cdot\floor{n/6}$ elements are less or equal and at least the same number of elements are greater or equal than the pivot~-- thus, always the larger part of the partitioned array can be sorted with Mergesort and the recursion takes place on the smaller part.
The advantage of this method is that the median-of-medians algorithm is applied to an array of size only $n/3$ instead of $n$ (with the cost of introducing a small overhead for finding the $n/3$ medians of three)~-- giving less weight to its big constant for the linear number of comparisons. 
We call this algorithm \emph{basic MoMQuickMergesort} (\bMQMS).

  For the median-of-medians algorithm, we use the repeated step method as described in \prettyref{sec:mom}. Notice that for the number of comparisons the worst case for MoMQuickMergesort happens if the pivot is exactly the median since this gives the most weight on the ``slow'' median-of-medians algorithm.
 Thus, the total number $T_{\bMQMS}(n)$ of comparisons of MoMQuickMergesort in the worst case to sort $n$ elements is bounded by
\begin{align*}T_{\bMQMS}(n) &\leq T_{\bMQMS}\left(\frac{n}{2}\right) + T_{\MS}\left(\frac{n}{2}\right) + T_{\MOM}\left(\frac{n}{3}\right) + 3\cdot\frac{n}{3}  + \frac{2}{3}n + \Oh(1)\end{align*}
where $T_{\MS}(n)$ is the number of comparisons of Mergesort and $T_{\MOM}(n)$ the number of comparisons of the median-of-medians algorithm. The $3\cdot\frac{n}{3}$-term comes from finding $n/3$ medians of three elements, the $2n/3$ comparisons from partitioning the remaining elements (after finding the pivot, the correct side of the partition is known for $n/3$ elements).

By \prettyref{lem:mom} we have $T_{\MOM}(n) \leq 20n + \Oh(n^{0.8})$ and by \cite{weisstein} we have $T_{\MS}(n) \leq n\log n - 0.91n + 1$. Thus, we can use \prettyref{lem:recurrence} to resolve the recurrence, which proves (notice that for every comparison there is only a constant number of other operations):

\begin{theorem}\label{thm:wc_basic}
Basic MoMQuickMergesort (\bMQMS) runs in $\Oh(n \log n)$ time and performs at most $n \log n + 13.8n + \Oh(n^{0.8})$ comparisons.
\end{theorem}

\subsection{Improved version}\label{sec:impQMS}

In \cite{Reinhardt92}, Reinhardt describes how to merge two subsequent sequences in an array using additional space for only half the number of elements in one of the two sequences. The additional space should be located in front or after the two sequences. To be more precise, assume we are given an array $A$ with positions $A[1, \dots , t]$ being empty or containing dummy elements (to simplify the description, we assume the first case), $A[t+1, \dots, t + \ell ]$ and $A[t + \ell+1,\dots, t + \ell + r ]$ containing two sorted sequences. We wish to merge the two sequences into the space $A[1, \dots, \ell +r]$ (so that $A[\ell+ r+1, \dots, t + \ell + r]$ becomes empty).
 We require that $r/2 \leq t < r $.

 First we start from the left merging the two sequences into the empty space until there remains no empty space between the last element of the already merged part and the first element of the left sequence (first step in \prettyref{fig:merge5}). At this point, we know that at least $t$ elements of the right sequence have been introduced into the merged part (because when introducing elements from the left part, the distance between the last element in the already merged part and the first element in the left part does not decrease). Thus, the positions $t + \ell  + 1$ through $\ell +2t$ are empty now. Since $\ell + t + 1 \leq \ell + r \leq \ell + 2t$, in particular, $A[\ell + r]$ is empty now.  
 Therefore, we can start merging the two sequences right-to-left into the now empty space (where the right-most element is moved to position $A[\ell + r]$~-- see the second step in \prettyref{fig:merge5}). Once the empty space is filled, we know that all elements from the right part have been inserted, so $A[1, \dots, \ell + r]$ is sorted and  $A[\ell+ r+1, \dots, t + \ell + r]$ is empty (last step in \prettyref{fig:merge5}).

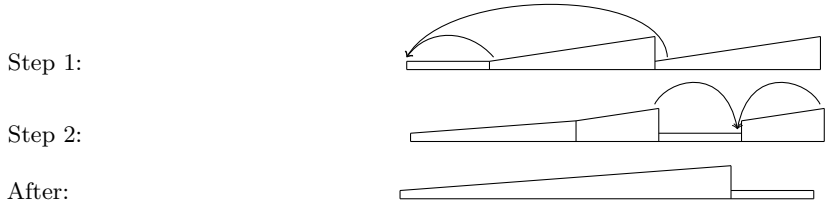
\begin{figure}[t]

\small
Step 1:\hfil
	\begin{scriptsize}
		\begin{tikzpicture}[scale = 0.55]
		
		\draw(0,0) -- (10,0);
		\draw(0,0.2) -- (2,0.2);
		\draw(2,0.2) -- (6,0.8);
		\draw(0,0) -- (0,0.2);
		\draw(6,0.2) -- (10,0.8);

		\draw[->](6.3,.3) ..controls(6.2,2) and (1, 2 ).. (0,.3);
		\draw[->](2.1,.3) ..controls(1.5,1) and (.5, 1 ).. (0,.3);
		
		\draw(2,0) -- (2,0.2);
		\draw(6,0) -- (6,0.8);
		\draw(10,0) -- (10,0.8);
		\end{tikzpicture}
	\end{scriptsize}
\hfill

\vspace{-2mm}
Step 2:\hfil
	\begin{scriptsize}
		\begin{tikzpicture}[scale = 0.55]
		
		\draw(0,0) -- (10,0);
		\draw(0,0.2) -- (4,0.5);
		\draw(4,0.5) -- (6,0.8);
		\draw(6,0.2) -- (8,0.2);
		\draw(0,0) -- (0,0.2);
		\draw(8,0.5) -- (10,0.8);

	\draw[->](5.9,.9) ..controls(6.2,1.5) and (7.6, 1.9 ).. (7.9,.3);
	\draw[->](9.9,.9) ..controls(9.6,1.5) and (8.2, 1.9 ).. (7.9,.3);
		
		\draw(4,0) -- (4,0.5);
		\draw(6,0) -- (6,0.8);
		\draw(8,0) -- (8,0.5);
		\draw(10,0) -- (10,0.8);
		\end{tikzpicture}
	\end{scriptsize}

\medskip

After:\hfil	
\begin{scriptsize}
		\begin{tikzpicture}[scale = 0.55]
		
		\draw(0,0) -- (10,0);
		\draw(0,0.2) -- (8,0.8);
		\draw(8,0.2) -- (10,0.2);

		\draw(0,0) -- (0,0.2);
		\draw(8,0) -- (8,0.8);
		\draw(10,0) -- (10,0.2);
		\end{tikzpicture}
	\end{scriptsize}
\hfill

	\caption{\small In the first step the two sequences are merged starting with the smallest elements until the empty space is filled. Then there is enough empty space to merge the sequences from the right into its final position.}\label{fig:merge5}
\end{figure}

When choosing $\ell=r$ (in order to have a balanced merging and so an optimal number of comparisons), we need one fifth of the array as temporary space. Moreover, by allowing a slightly imbalanced merge we can also tolerate slightly less temporary space. In the case that the temporary space is large ($t \geq r$), we apply the merging scheme from \prettyref{sec:qms}. The situation where the temporary space is located after the two sorted sequences is handled symmetrically (note that this changes the requirement to $\ell/2 \leq t < \ell $).

By applying this merging method in MoMQuickMergesort, we can use pivots having much weaker guarantees: instead of one third, we need only one fifth of the elements being less (resp.\ greater) than the pivot.
 We can find such pivots by applying an idea similar to the repeated step method for the median-of-medians algorithm: first we group into blocks of fifteen elements and compute the pseudomedians of each group. Then, the pivot is selected as median of these pseudomedians; it is computed using the median-of-medians algorithm. This guarantees that at least $2\cdot 3 \cdot \floor{\frac{n}{3\cdot 5 \cdot 2}} \approx \frac{n}{5}$ elements are less than or equal to (resp.\ greater than or equal to) the pivot. Computing the pseudomedian of 15 elements requires 22 comparisons (five times three comparisons for the medians of three and then seven comparisons for the median of five). After that, partitioning requires $14/15n$ comparisons.
Since still in any case the larger half can be sorted with Mergesort, we get the recurrence (we call this algorithm \emph{MoMQuickMergesort} (\MQMS))
\begin{align*}
T_{\MQMS}(n) &\leq T_{\MQMS}(n/2) + T_{\MS}(n/2) + T_{\MOM}(n/15) + \frac{22}{15}n + \frac{14}{15}n + \Oh(1)\\
&\leq  T_{\MQMS}(n/2) + \frac{n}{2}\log(n/2) - \frac{0.91 n}{2}+ \frac{20}{15}n + \frac{36}{15}n + \Oh(n^{0.8})\\
& \leq  n \log n - 0.91 n - 2n + \frac{112}{15}n + \Oh(n^{0.8})\tag{by \prettyref{lem:recurrence}} 
\end{align*}
This proves:
\begin{theorem}\label{thm:wc}
	 MoMQuickMergesort (\MQMS) runs in $\Oh(n \log n)$ time and performs at most $n \log n + 4.57n + \Oh(n^{0.8})$ comparisons.
\end{theorem}

Notice that when computing the median of pseudomedians of fifteen elements, in the worst case approximately the same effort goes into the calculation of the pseudomedians and into the median-of-medians algorithm. This indicates that it is an efficient method for finding a pivot with the guarantee that one fifth are greater or equal (resp.\ less or equal).

\subsection{Undersampling}\label{sec:small_samples}

In \cite{Alexandrescu17} Alexandrescu selects pivots for the median-of-medians algorithm not as medians of medians of the whole array but only of $n/\phi$ elements where $\phi$ is some large constant (similar as in \cite{Kurosawa16} for Quicksort). While this improves the average case considerably and still gives a linear time algorithm, the hidden constant for the worst case is large. In this section we follow the idea to a certain extent without loosing a good worst-case bound.

As already mentioned in \prettyref{sec:impQMS}, Reinhardt's merging procedure \cite{Reinhardt92} works also with less than one fifth of the whole array as temporary space if we do not require to merge sequences of equal length.
Thus, we can allow the pivot to be even further off the median~-- 
with the cost of making the Mergesort part more expensive due to imbalanced merging. 
For $\theta \geq 1$ we describe a variant $\uMQMS{\theta}$ of MoMQuickMergesort using only $ n/\theta$ elements for sampling the pivot. Before we analyze this variant, let us look at the costs of Mergesort with imbalanced merging:
in order to apply Reinhardt's merging algorithm, we need that one part is at most twice the length of the temporary space. We always apply linear merging (no binary insertion) meaning that merging two sequences of combined length $n$ costs at most $n-1$ comparisons. Thus, we get the following estimate for the worst case number of comparisons $ T_{\MS,b}(n,m)$ of Mergesort where $n$ is the number of elements to sort and $m$ is the temporary space (= ``buffer''):
\begin{align*}
 T_{\MS,b}(n,m) &\leq \!\begin{cases}
T_{\MS}(n) & \text{if } n \leq 4m\\
n +  T_{\MS,b}(n - 2m,m) + T_{\MS}(2m)&\text{otherwise}.
\end{cases}
\end{align*}
If $n > 2m$ (otherwise, there is nothing to do), this means
\begin{align}
 \nonumber T_{\MS,b}(n,m) 	&\leq \left(\ceil{\frac{n}{2m}} - 2\right) \cdot T(2m)  + T\left(n - \left(\ceil{\frac{n}{2m}} - 2\right)\cdot 2m\right)
+ \!\!\!\sum_{i=0}^{\ceil{\frac{n}{2m}} - 3} (n - 2im)\\ 
 \begin{split}
	&= \left(\ceil{\frac{n}{2m}} - 2\right) \cdot T(2m) + T\left(n - \left(\ceil{\frac{n}{2m}} - 2\right)\cdot 2m\right)\\ 
&\qquad\;  + n \cdot \left(\ceil{\frac{n}{2m}} - 2\right)  - m \cdot  \left(\ceil{\frac{n}{2m}} - 2\right) \left(\ceil{\frac{n}{2m}} - 3\right).\label{eq:TMSb}
 \end{split}
\end{align}
For a moment let us assume that $\frac{n}{2m} = \ell \in \Z$ (with $\ell \geq 1 $). In this case we have 
\begin{align*}
 T_{\MS,b}\left(n,\frac{n}{2\ell}\right) &\leq (\ell- 2) \cdot T\left(\frac{n}{\ell}\right) + T\Bigl(\frac{2n}{\ell}\Bigr) + n \cdot (\ell - 2) - \tfrac{n}{2\ell} \cdot  (\ell - 2) (\ell - 3)\\
%		&= (\ell- 2) \cdot T(n/\ell) + T\left(\tfrac{2n}{\ell}\right)  + n \cdot (\ell - 2) \cdot \left(1/2 + \tfrac{3}{2\ell} \right)\\
		&\leq (\ell- 2) \cdot \tfrac{n}{\ell} \cdot \left(\log\left(\tfrac{n}{\ell}\right) - \kappa \right) 
		+ \tfrac{2n}{\ell} \cdot \left(\log \left(\tfrac{2n}{\ell}\right) - \kappa \right) + n \cdot (\ell - 2) \cdot\left(\tfrac{1}{2} + \tfrac{3}{2\ell} \right)\\
		&\leq n \log n 		+ n \cdot f(\ell) 
\end{align*}
 for $n$ large enough where $f: \R_{>0} \to \R$ is defined by
\begin{align*}
f(\ell) = \begin{cases}
- \kappa   - \log \ell + \ell/2  +1/2 - 1/\ell&\text{for } \ell \geq 2\\
 -\kappa&\text{otherwise} 
\end{cases}
\end{align*}
and $n\log n - \kappa n$ is a bound for the number of comparisons of Mergesort for $n$ large enough ($\kappa \approx 0.91$ by \cite{weisstein}).
Now, for arbitrary $m$ we can use $f$ as an approximation, which turns out to be quite precise: 
\begin{lemma}\label{lem:error_bound}
	Let $f$ be defined as above and write $\ceil{\frac{n}{2m}} = \frac{n}{2m} + \xi$. Then for $m$ and $n$ large enough we have
	 \begin{align*}
	 T_{\MS,b}(n&,m)  \leq  
	 n \log n + n \cdot f\left(\frac{n}{2m}\right) + m\cdot \epsilon(\xi)
	 \end{align*}
	  where $\epsilon(\xi) = \max \oneset{0, 5\xi - 4 + (4 - 2\xi)  \log (2 - \xi) - 
            \xi^2} \\ \leq  0.015 =:\epsilon$.
\end{lemma}
\begin{proof}
Let $m$ be large enough such that $T_{\MS}(m) \leq m \log m - \kappa m$. If for $n \leq 2m$ we have $T_{\MS,b}\left(n,m\right) = T_{\MS}\left(n\right)$ and so the lemma holds. Now let $n > 2m$.	
By \prettyref{eq:TMSb} we obtain 
\begin{align}
 T_{\MS,b}\left(n,m\right) \leq
&\left(\frac{n}{2m} + \xi -2\right)T_{\MS}(2m)+ T_{\MS}((2-\xi) 2m) \label{eq:part1}\\
&\label{eq:part2}\quad + n \cdot \left(\frac{n}{2m} + \xi - 2\right)  -  m \cdot  \left(\frac{n}{2m}+ \xi - 2\right) \left(\frac{n}{2m} + \xi - 3\right).
\end{align}
We examine the two terms \prettyref{eq:part1} and \prettyref{eq:part2} separately using $T_{\MS}(n) \leq n \log n - \kappa n$:	
\begin{align*}\allowdisplaybreaks
\prettyref{eq:part1} &= \left(\frac{n}{2m} + \xi -2\right)T_{\MS}(2m)+ T_{\MS}((2-\xi) 2m) \\ 
&\leq \left(\frac{n}{2m} + \xi -2\right)2m \left(\log 2m - \kappa \right ) +(2-\xi) 2m \bigl(\log ((2-\xi) 2m) - \kappa \bigr) \\ 
&= (n \log 2m - \kappa n)+   \left(\xi -2\right)2m(\log 2m - \kappa )  + (2-\xi) 2m \bigl(\log ((2-\xi) 2m) - \kappa \bigr) \\ 
&= (n \log n - \kappa n) - n \log \left(\frac{n}{2m}\right)  + \left(2-\xi\right)2m  \cdot \bigl( (\log ((2-\xi) 2m) - \kappa ) - (\log 2m - \kappa ) \bigr)\\
&= (n \log n - \kappa n) - n \log \left(\frac{n}{2m}\right)  + \left(2-\xi\right)2m \log (2-\xi)   
\intertext{and}
\prettyref{eq:part2}&=  n \cdot \left(\frac{n}{2m} + \xi - 2\right)  - m \cdot  \left(\frac{n}{2m}+ \xi - 2\right) \left(\frac{n}{2m} + \xi - 3\right)\\ 
&=  n \cdot \left(\frac{n}{2m}  - 2\right) - m \cdot  \left(\frac{n}{2m} - 2\right) \left(\frac{n}{2m}  - 3\right)  + n\xi - m \left(\xi \left(\frac{n}{2m}  - 3\right) + \left(\frac{n}{2m} - 2\right) \xi + \xi ^2\right)\\
&=  n \cdot \left(\frac{n}{2m}  - 2 - \frac{m}{n} \cdot   \left(\left(\frac{n}{2m}\right)^2 - 5\frac{n}{2m} + 6\right)\right)  + m \left(5 \xi - \xi ^2\right)\\	
&=  n \cdot \left(\frac{n}{2\cdot 2m}  + \frac{1}{2} - 3 \cdot\frac{2m}{n}\right)  + m \left(5 \xi - \xi ^2\right).	
\end{align*}
Thus, 	
\begin{align*}
T_{\MS,b}(n,m) -  \left(n \log n + n \cdot f\left(\frac{n}{2m}\right)\right) & \leq \prettyref{eq:part1} + \prettyref{eq:part2}  -  \left(n \log n + n \cdot f\left(\frac{n}{2m}\right) \right) \\ &\leq \left(2-\xi\right)2m \log (2-\xi) + m \left(5 \xi - \xi ^2\right) - 4m.
\end{align*}
	This completes the proof of \prettyref{lem:error_bound}.
\end{proof}

For selecting the pivot in QuickMergesort, we apply the procedure of \prettyref{sec:impQMS} to $n/\theta$ elements (for some parameter $\theta \in \R$, $\theta\geq 1$): we select $n/\theta$ elements from the array, group them into groups of fifteen elements, compute the pseudomedian of each group, and take the median of those pseudomedians as pivot.  We call this algorithm MoMQuickMergesort with \emph{undersampling factor} $\theta$ (\uMQMS{\theta}). Note that $\uMQMS{1} = \MQMS$. For its worst case number of comparisons we have 
\begin{align*}
 T_{\uMQMS{\theta}}(n) &\leq \max_{\frac{1}{5 \theta }\leq \alpha \leq \frac{1}{2}} T_{\uMQMS{\theta}}\left(\alpha n\right) + T_{\MS,b}\left(n \left (1 - \alpha\right),\alpha n\right)  \\&\qquad +  \frac{22}{15\theta}n  + \frac{20}{15\theta}n +\left(1 - \frac{1}{15\theta}\right)n + \Oh(n^{0.8}) 
\end{align*}
where the $\frac{22}{15\theta}n$ is for finding the pseudomedians of fifteen, the $\frac{20}{15\theta}n + \Oh(n^{0.8})\vphantom{k^{k^k}}$ is for the median-of-medians algorithm called on $\frac{n}{15\theta}$ elements and $\left(1 - \frac{1}{15\theta}\right)n$ is for partitioning the remaining elements. 
Now we plug in the bound of \prettyref{lem:error_bound} for $ T_{\MS,b}(n,m) $ with $\ell = \frac{1-\alpha}{2\alpha}$ and apply \prettyref{lem:recurrence}: 
\begin{align*}
T_{\uMQMS{\theta}}(n) 
&\leq  \max_{\frac{1}{5 \theta }\leq \alpha \leq \frac{1}{2}} T_{\uMQMS{\theta}}\left(\alpha n\right) +  (1-\alpha) n\log ((1-\alpha) n)\\&\qquad + (1-\alpha) n  \left( f\left(\frac{1-\alpha}{2\alpha}\right) + \epsilon\right) + n \cdot \left( 1 + \frac{41}{15\theta}  \right) + \Oh(n^{0.8})\\
&\leq   n\log n +  n \cdot \max_{\frac{1}{5 \theta }\leq \alpha \leq \frac{1}{2}} g(\alpha, \theta) + \Oh(n^{0.8})
\end{align*}
for 
\begin{align*}
g(\alpha, \theta) &= \frac{\alpha\log (\alpha)}{1-\alpha} +  \log \left( 1 - \alpha\right)  +   f\left(\frac{1-\alpha}{2\alpha}\right)  + \frac{1}{1-\alpha} \cdot \left( 1 + \frac{41}{15\theta}  \right)+ \epsilon.
\end{align*}

In order to find a good undersampling factor, we wish to find a value for $\theta$ minimizing $\max_{\frac{1}{5 \theta }\leq \alpha \leq \frac{1}{2}} g(\alpha, \theta)$. While  we do not have a formal proof, intuitively the maximum should be either reached for $\alpha = 1/2$ (if $\theta$ is small) or for $\alpha = 1/(5\theta)$ (if $\theta$ is large)~-- see \prettyref{fig:undersampling-fixed-theta} for a special value of $\theta$.  Moreover, notice that we are dealing with an upper bound on $T_{\uMQMS{\theta}}(n)$  only (with a small error due to \prettyref{lem:error_bound} and the bound $n\log n - \kappa n$ for $T_{\MS}(n)$), so even if we could find the $\theta$ which minimizes $\max_{\frac{1}{5 \theta }\leq \alpha \leq \frac{1}{2}} g(\alpha, \theta)$, this $\theta$ might not be optimal.

We proceed as follows: first, we compute the point $\theta_{\mathrm{opt}}$ where the two curves in \prettyref{fig:undersampling-theory} intersect. For this particular value of $\theta$, we then show that indeed  $\max_{\frac{1}{5 \theta }\leq \alpha \leq \frac{1}{2}} g(\alpha, \theta) = g(1/2,\theta)$. Since $\theta  \mapsto g(1/2, \theta)$ is monotonically decreasing (this is obvious) and $\theta \mapsto g(1/(5\theta), \theta)$ is monotonically increasing for $\theta \geq 2.13$ (verified numerically), this together shows that $\max_{\frac{1}{5 \theta }\leq \alpha \leq \frac{1}{2}} g(\alpha, \theta)$ is minimized at the intersection point.

\begin{figure}[htb]
	\vspace{-6mm}
	\begin{center}
		\scalebox{0.82}{\hspace{-1.5mm}\input{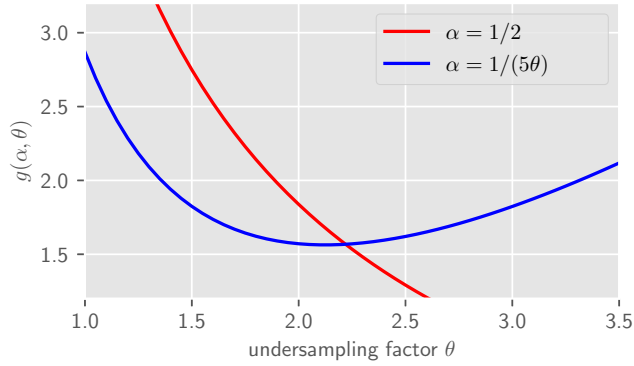}}
	\end{center}
	\vspace{-3mm}
	\caption{\small $\theta  \mapsto g(\alpha, \theta)$ for $\alpha= 1/2$ and $\alpha=1/(5\theta)$.}\label{fig:undersampling-theory}
\end{figure}

We compute the intersection point numerically as $\theta_{\mathrm{opt}} \approx 2.219695$. For $\theta_{\mathrm{opt}}$ we verify (using Wolfram$|$Alpha \cite{wolframalpha}), that the maximum $\max_{\frac{1}{5 \theta }\leq \alpha \leq \frac{1}{2}} g(\alpha, \theta)$ is attained at $\alpha = 1/(5\theta_{\mathrm{opt}})$ and that $g(1/(5\theta_{\mathrm{opt}}), \theta_{\mathrm{opt}}) \approx 1.56780$ and $g(1/2, \theta_{\mathrm{opt}}) \approx 1.56780$. Thus, we have established the optimality of $\theta_{\mathrm{opt}}$ even though we have not computed $g(\alpha, \theta)$ for $\alpha \not\in \oneset{1/5, 1/(5\theta)}$ and $\theta \neq \theta_{\mathrm{opt}}$. (In the mathoverflow question \cite{mathoverflowP18}, this value is verified analytically~-- notice that there $g(\alpha, \theta)$ is slightly different giving a different $\theta$.)

For implementation reasons we want $\theta$ to be a multiple of $1/30$. 
Therefore, we propose $\theta = 11/5$~-- a choice which is only slightly smaller than the optimal value and confirmed experimentally (\prettyref{fig:undersampling}).
Again for this fixed $\theta$, we verify that indeed the maximum is at $\alpha = 1/2$ and that $g(1/(5\theta), \theta) \approx 1.57$ and  $g(1/2, \theta) \approx 1.59$, see \prettyref{fig:undersampling-fixed-theta}. Thus, up to the small difference $0.02$, we know that $\theta=11/5$ is optimal.

\begin{figure}[htb]
	\vspace{-6mm}
	\begin{center}
		\scalebox{0.82}{\hspace{1.5mm}\input{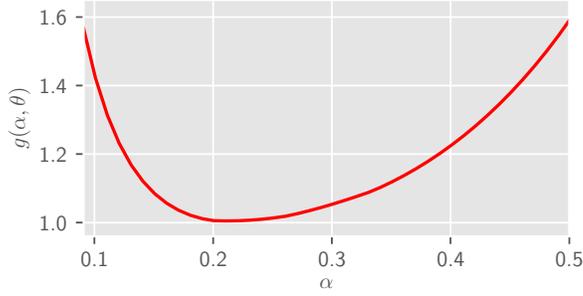}}
	\end{center}
	\vspace{-3mm}
	\caption{\small $\alpha  \mapsto g(\alpha, \theta)$ for $\theta= 11/5$ with $\alpha \in [1/(5\theta), 1/2]$ reaches its maximum $\approx 1.59$ for $\alpha = 1/2$. }\label{fig:undersampling-fixed-theta}
\end{figure}

For this fixed value of $\theta = 11/5$ we have thus computed  $\max_{\frac{1}{5 \theta }\leq \alpha \leq \frac{1}{2}} g(\alpha, \theta) \leq 1.59$, which in turn gives us a bound on  $T_{\uMQMS{11/5}}(n)$.

\begin{theorem}\label{thm:wc_undersampling}
	MoMQuickMergesort with undersampling factor $\theta = 11/5$ ($\uMQMS{11/5}$) runs in $\Oh(n \log n)$ time and performs at most $n \log n + 1.59n + \Oh(n^{0.8})$ comparisons.
\end{theorem}

\subsection{Heuristic estimate of the average case}\label{sec:average}

It is hard to calculate an exact average case since at none but the first stage during the execution of the algorithm we are dealing with random inputs. We still estimate the average case by assuming that all intermediate arrays are random and applying some more heuristic arguments. 

\smallskip
\paragraph{Average of the median-of-medians algorithm.} On average we can expect that the pivot returned from the median-of-medians procedure is very close to an actual median, which gives us an easy recurrence showing that $T_{\mathrm{av, MoM}}(n) \approx 40/7n$. 
However, we have to take adaptive pivot selection into account. The first pivot is the $n/2 \pm o(n)$-th element with very high probability. Thus, the recursive call is on $n/2 + o(n)$ elements with $k \in o(n)$ (or $k=n/2$~-- by symmetry we assume the first case). Due to adaptive pivot selection, the array will be also split in a left part of size $o(n)$ (with the element we are looking for in it~-- this is guaranteed even in the worst case) and a larger right part. This is because an $o(n)$ order element of the $n/15$ pseudomedians of fifteen is also an $o(n)$ order elements of the whole array. Thus, all successive recursive calls will be made on arrays of size $o(n)$. We denote the average number of comparisons of the median-of-median algorithm recursing on an array of size $o(n)$ as $T_{\mathrm{av, MoM}}^{\mathrm{noRec}}(n)$.

We also have to take the recursive calls for pivot selection  into account.
The first pivot is the median of the sample; thus, the same reasoning as for $T_{\mathrm{av, MoM}}(n)$ applies. The second pivot is an element of order $o(n)$ out of $n/18 + o(n)$ elements~-- so we are in the situation of $T_{\mathrm{av, MoM}}^{\mathrm{noRec}}(n)$. Thus, we get
\begin{align*}
T_{\mathrm{av, MoM}}(n) &= T_{\mathrm{av, MoM}}^{\mathrm{noRec}}\left(\frac{n}{2}\right) + T_{\mathrm{av, MoM}}\left(\frac{n}{9}\right) + \frac{20n}{9}
\intertext{and}
T_{\mathrm{av, MoM}}^{\mathrm{noRec}}(n) &=  T_{\mathrm{av, MoM}}^{\mathrm{noRec}}(\tfrac{n}{9}) +\tfrac{20n}{9} + o(n).
\end{align*}
Hence, by \prettyref{lem:recurrence}, we obtain $T_{\mathrm{av, MoM}}^{\mathrm{noRec}}(n) = \tfrac{20n}{9}\cdot\frac{9}{8} + o(n) = \tfrac{5n}{2}+ o(n)$ and \begin{align*}
T_{\mathrm{av, MoM}}(n) &= T_{\mathrm{av, MoM}}(n/9) + \frac{20n}{9} + \frac{5n}{4}  + o(n)\\ &= \frac{125}{32}n + o(n) \leq 4n + o(n).
\end{align*}

\paragraph{Average of MoMQuickMergesort.}
As for the median-of-medians algorithm, we can expect that the pivot in MoMQuickMergesort is always very close to the median. Using the bound for the adaptive version of the median-of-medians algorithm, we obtain 
\begin{align*}\allowdisplaybreaks
	T_{\mathrm{av},\uMQMS{\theta}}(n) &= T_{\mathrm{av},\uMQMS{\theta}}(n/2) + \frac{n}{2}\log(n/2) - \frac{1.24 n}{2}+ \frac{22}{15\theta}n + \frac{4}{15\theta}n + \frac{15\theta - 1}{15\theta}n + o(n)\\
	&\leq n\log n + n \cdot\left(- 1.24 +  \frac{10}{3\theta} \right)+ o(n).
\end{align*}
by \prettyref{lem:recurrence} (here the $4n/15\theta$ is for the average case of the median-of-medians algorithm, the other terms as before). This yields
\[T_{\mathrm{av}, \MQMS}(n) \leq n\log n + 2.094 n + o(n)\] (for $\theta=1$). For our proposed $\theta=\frac{11}{5}=2.2$ we have
\[T_{\mathrm{av},\uMQMS{11/5}}(n) \leq n\log n + 0.275 n + o(n).\]

\subsection{Hybrid algorithms}\label{sec:hybrid}

In order to achieve an even better average case, we can apply a trick similar to {Introsort} \cite{Mus97}. Be aware, however, that this deteriorates the worst case slightly. 
We fix some small $\delta >0$. The algorithms starts by executing QuickMergesort with median of three pivot selection. Whenever the pivot is
contained in the interval $\left[\delta n, (1-\delta)n \right]$, the next pivot is selected again as median of three, otherwise according to \prettyref{sec:small_samples} (as median of pseudomedians of $n/\theta$ elements)~-- for the following pivots it switches back to median of 3. When choosing  $\delta$ not too small, the worst case number of comparisons will be only approximately $2n$ more than of MoMQuickMergesort with undersampling (because in the worst case before every partitioning step according to MoMQuickMergesort with undersampling, there will be one partitioning step with median-of-3 using $n$ comparisons), while the average is almost as QuickMergesort with median-of-3. We use $\delta = 1/16$.
We call this algorithm hybrid QuickMergesort (HQMS). 

Another possibility for a hybrid algorithm is to use MoMQuickMergesort (with undersampling) instead of Heapsort as a worst-case stopper for Introsort. We test both variants in our experiments.

\subsection{Summary of algorithms}
For the reader's convenience we provide a short summary of the different versions of MoMQuickMergesort and the results we obtained in \prettyref{tab:summary}.

\begin{table}[hbt]%
	\begin{center}
		{\small\renewcommand{\arraystretch}{1.4}	\begin{tabular}{|l|p{7.6cm}|p{3.7cm}|}
				\hline
				Acronym &  Algorithm  &  Results \\  \cline{2-3}
				\hline
		$\bMQMS$ 		&  basic MoMQuickMergesort 		 	
						&\prettyref{thm:wc_basic}:  $ \kappa_{\mathrm{wc}} \leq 13.8$  \\
		$\MQMS$	 		& 	MoMQuickMergesort\newline (uses Reinhardt's merging with balanced merges) 			 
						&\prettyref{thm:wc}: $ \kappa_{\mathrm{wc}} \leq 4.57$, \newline\prettyref{sec:average}:  $ \kappa_{\mathrm{ac}} \approx 2.094$\\
		$\uMQMS{11/5}$ 	&	MoMQuickMergesort with undersampling factor $11/5$\newline (uses Reinhardt's merging with imbalanced merges)
		 				&\prettyref{thm:wc_undersampling}: $ \kappa_{\mathrm{wc}} \leq 1.59$, \newline \prettyref{sec:average}: $ \kappa_{\mathrm{ac}} \approx 0.275$\\
		HQMS 			& 	hybrid QuickMergesort	(combines median-of-3 QuickMergesort and $\uMQMS{11/5}$) 
						& \prettyref{sec:hybrid}:  $ \kappa_{\mathrm{wc}} \leq 3.58$ for $\delta$ large enough, \newline $ \kappa_{\mathrm{ac}}$ smaller
			 \\
				\hline
			\end{tabular}
		}
	\end{center}	
	\vspace{-3mm}
	\caption{\small Overview over the algorithms in this paper. For the worst case number of comparisons $n\log n - \kappa_{\mathrm{wc}}n + \Oh(n^{0.8})$ and average case  of roughly $n\log n - \kappa_{\mathrm{ac}}n$ the results on $ \kappa_{\mathrm{wc}}$ and $\kappa_{\mathrm{ac}}$ are shown. The average cases are only heuristic estimates.}\label{tab:summary}
\end{table}%

\section{Experiments}\label{sec:experiments}

\paragraph{Experimental setup.} We ran thorough experiments with implementations in C++ with different kinds of input
permutations. The experiments are run on an Intel Core i5-2500K CPU (3.30GHz, 4 cores,
32KB L1 instruction and data cache, 256KB L2 cache per core and 6MB L3
shared cache) with 16GB RAM and operating system Ubuntu Linux 64bit
version 14.04.4.  We used GNU's \texttt{g++} (4.8.4); optimized with
flags \texttt{-O3 -march=native}.
For time measurements, we used \texttt{std\dd chrono\dd high\_resolution\_clock}, for generating random inputs, the Mersenne
Twister pseudo-random generator \texttt{std:$\!$:mt19937}. All time
measurements were repeated with the same 100 deterministically chosen
seeds~-- the displayed numbers are the averages of these 100 runs. 
Moreover, for each time measurement, at least 128MB of data were sorted~-- if the array size is smaller, then for this time
measurement several arrays have been sorted and the total
elapsed time measured.
If not specified explicitly, all experiments were conducted with 32-bit integers.

\smallskip
\paragraph{Implementation details.}
The code of our implementation of MoMQuickMergesort as well as the other algorithms and our running time experiments is available at \url{https://github.com/weissan/QuickXsort}.
 In our implementation of MoMQuickMergesort, we use the merging procedure from \cite{ElmasryKS12}, which avoids branch mispredictions. We use the partitioner from the libstdc++ implementation of \stdsort. For the running time experiments, base cases up to 42 elements are sorted with Insertionsort. For the comparison measurements Mergesort is used down to size one arrays.

\smallskip
\paragraph{Simulation of a worst case.}
In order to experimentally confirm our worst case bounds for MoMQuickMergesort, we simulate a worst case.  Be aware that it is not even clear whether in reality there are input permutations where the bounds for the worst case of \prettyref{sec:momqms} are tight since when selecting pivots the array is already pre-sorted in a particular way (which is hard to understand for a thorough analysis). Actually in \cite{ChenD15} it is conjectured that similar bounds for different variants of the median-of-medians algorithm are not tight. Therefore, we cannot test the worst-case by designing particularly bad inputs. Nevertheless, we can simulate a worst-case scenario where every pivot is chosen the worst way possible (according to the theoretical analysis). More precisely, the simulation of the worst case comprises the following aspects:
\begin{itemize}
	\item 
%i) 
For computing the $k$-th element of a small array (up to 30 elements) we additionally sort it with Heapsort. This is because our implementation uses Introselect (\stdnth) for arrays of size up to 30.
	\item 
%ii) 
When measuring comparisons, we perform a random shuffle before every call to Mergesort. As the average case of Mergesort is close to its worst case (up to approximately $0.34n$), this gives a fairly well approximation of the worst case. For measuring time we apply some simplified shuffling method, which shuffles only few positions in the array. 
	\item 
%iii)
 In the median-of-medians algorithm, we do not use the pivot selected by recursive calls, but use \stdnth to find the worst pivot the recursive procedure could possibly select. We do not count comparisons incurred by \stdnth. This is the main contribution to the worst case.
	\item 
%iv) 
As pivot for QuickMergesort (the basic and improved variant) we always use the real median (this is actually the worst possibility as the recursive call of QuickMergesort is guaranteed to be on the smaller half and Mergesort is not slower than QuickMergesort). In the version with undersampling we use the most extreme pivot (since this is worse than the median).
\item We also make 100 measurements for each data point. When counting comparisons, we take the maximum over all runs instead of the mean. However, this makes only a negligible difference (as the small standard deviation in \prettyref{tab:qmsComp} suggests).  When measuring running times we still take the mean since the maximum reflects only the large standard deviation of Quickselect (\stdnth), which we use to find bad pivots.
\end{itemize}

The simulated worst cases are always drawn as dashed lines in the plots (except in \prettyref{fig:undersampling}).

\smallskip
\paragraph{Different undersampling factors.}
In \prettyref{fig:undersampling}, we compare the (simulated) worst-case number of comparisons for different undersampling factors $\theta$. The picture resembles the one in \prettyref{fig:undersampling-theory}. However, all numbers are around $0.4$ smaller than in \prettyref{fig:undersampling-theory} because we used the average case of Mergesort to simulate its worst case. Also, depending on the array size $n$, the point where the two curves for $\alpha = 1/2$ and $\alpha = 1/(5\theta)$ meet differs ($\alpha$ as in \prettyref{sec:small_samples}). 
Still the minimum is always achieved between 2.1 and 2.3 (recall that we have to take the maximum of the two curves for the same $n$)~-- confirming the calculations in \prettyref{sec:small_samples} and suggesting $\theta = 2.2$ as a good choice for further experiments.

\begin{figure}[htb]
	\vspace{-6mm}
	\begin{center}
		\scalebox{0.82}{\hspace{-1.5mm}\input{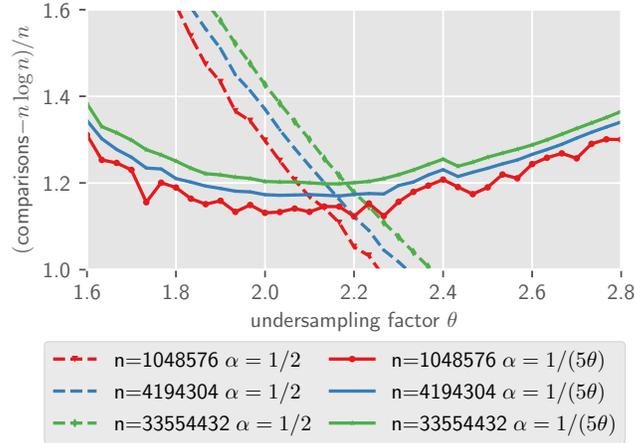}}
	\end{center}
	\vspace{12mm}
	\caption{\small Coefficient of the linear term of the number of comparisons in the simulated worst case for different undersampling factors.}\label{fig:undersampling}
\end{figure}

\begin{figure}[thb]
	\vspace{-10mm}
	\centering
	\scalebox{0.82}{\hspace{-2mm}\input{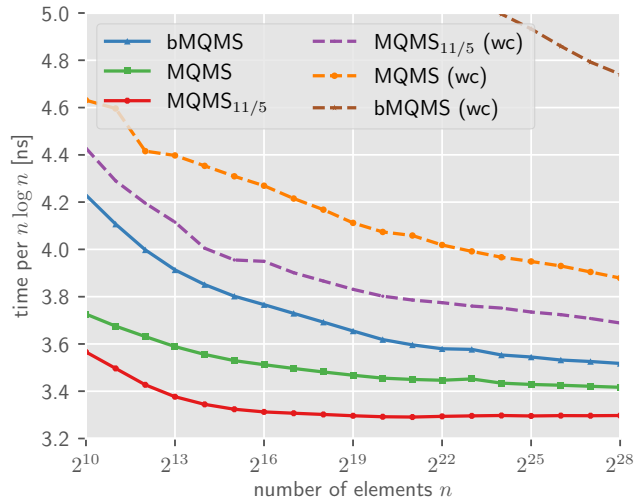}}
		\vspace{-2mm}
	\caption{\small Running times divided by $n \log n$ of different MoMQuickMergesort variants and their simulated worst cases.}\label{fig:qmsTime}
\end{figure}
\begin{figure}[htb]%
	\vspace{-8mm}
	\centering	\scalebox{0.82}{\hspace{-2mm}\input{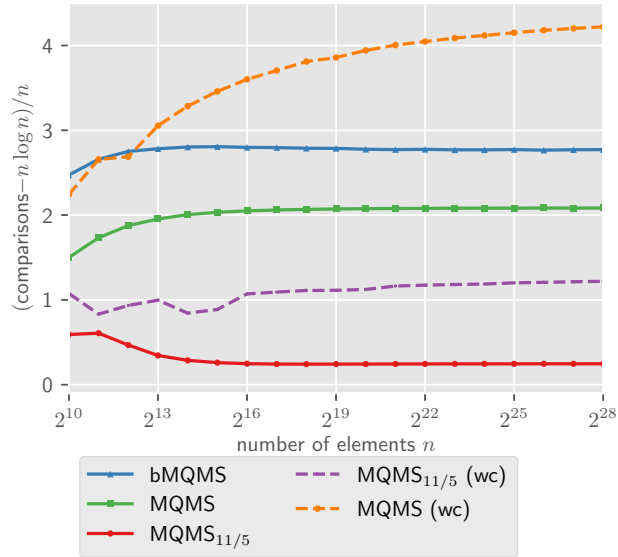}}
	\vspace{11mm}
	\caption{\small Number of comparisons (linear term) of different MoMQuickMergesort variants and their simulated worst cases. The worst case of \bMQMS\ is out of range.}\label{fig:qmsComp}
\end{figure}%

\smallskip
\paragraph{Comparison of different variants.}
In \prettyref{fig:qmsTime}, we compare the running times (divided by $n \log n$) of the different variants of MoMQuickMergesort including the simulated worst cases. We see that the version with undersampling is the fastest both in the average and worst case. Moreover, while in the average case the differences are rather small, in the worst case the improved versions are considerably better than the very basic variant. 

In \prettyref{fig:qmsComp} we count the number of comparisons of the different versions. The plot shows the coefficient of the linear term of the number of comparisons (i.\,e.\ the total number of comparisons minus $n \log n$ and then divided by $n$). \prettyref{tab:qmsComp} summarizes the results for $n=2^{28}$. 
We see that our theoretical estimates are close to the real values: for the average case, the difference is almost negligible; for the worst case, the gap is slightly larger because we use the average case of Mergesort as ``simulation'' for its worst case (notice that the difference between the average and our bound for the worst case is approximately $0.34n$). 
Moreover, the data suggest that for the worst case of $\bMQMS$ we would have to do experiment with even larger arrays in order to get a good estimate of the linear term of the number of comparisons.
Also we see that the actual number of comparisons approaches from below towards the theoretically estimated values~-- thus, the $\Oh(n^{0.8})$-terms in our estimates are  most likely negative. Notice however that, as remarked above, we do not know whether the bounds for the worst case are tight for real inputs.

\begin{table}[tb]%
	\begin{center}
		{\small\renewcommand{\arraystretch}{1.2}	\begin{tabular}{|l||p{2.2cm}|l|p{2.2cm}|l|}
			\hline
			\multirow{2}{*}{Algorithm\!\!\!}  & \multicolumn{2}{c|}{average case} & \multicolumn{2}{c|}{worst case}  \\  \cline{2-5}
			& exp.  &theo.& exp.  &theo.  \\
			\hline
			$\bMQMS$  		&$2.772$$\phantom{u}\pm 0.02$	& --  	&$13.05 $ $\phantom{u}\pm 0.17$ &13.8\\
			$\MQMS$ 		&$2.084$$\phantom{u}\pm 0.001 $&2.094	& $4.220$ $\phantom{u}\pm 0.007$ &4.57\\
			$\uMQMS{11/5}$ 	&$0.246$$\phantom{u}\pm 0.01 $&0.275	&$1.218$  $\phantom{u}\pm 0.011$ &1.59\\
			\hline
		\end{tabular}
	}
	\end{center}	
	\vspace{-3mm}
	\caption{\small Experimentally established linear term of the average and worst case (simulated) number of comparisons for MoMQuickMergesort for $n= 2^{28}$. The $\pm$ values are the standard deviations. The respective second columns show our theoretical estimates.}\label{tab:qmsComp}
\end{table}%

\begin{figure}[tb]%
	\vspace{-7mm}
	\begin{center}
		\scalebox{0.82}{\hspace{-2mm}\input{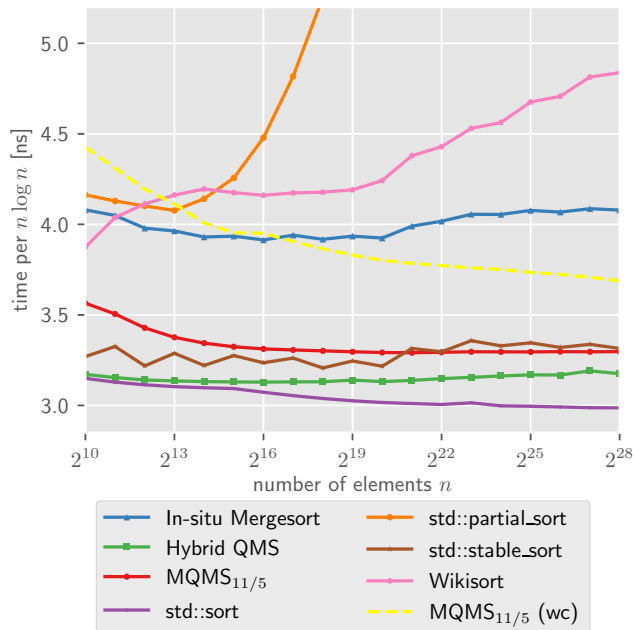}}
	\end{center}
	\vspace{12mm}	
	\caption{\small Running times of MoMQuickMergesort(average and simulated worst case), hybrid QMS and other algorithms for random permutations of 32-bit integers. Running times are divided by $n\log n$.}\label{fig:allAlgs}
\end{figure}%

\smallskip
\paragraph{Comparison with other algorithms.}
We conducted experiments comparing MoMQuickMergesort with the following other algorithms: Wikisort \cite{wikisort}, in-situ Mergesort \cite{ElmasryKS12}, \stdpartialsort (Heapsort), \stdstablesort (Mergesort) and \stdsort (Introsort), see \prettyref{fig:allAlgs}. For the latter three algorithms we use the libstdc++ implementations (from GCC version 4.8).
We also ran experiments with bottom-up Heapsort and Grailsort \cite{grailsort}, but omitted the results because these algorithms behave similar to Heapsort (resp.\ Wikisort). We see that MoMQuickMergesort with undersampling (\uMQMS{11/5}) performs better than all other algorithms except hybrid QuickMergesort and \stdsort. Moreover, for $n=2^{28}$ the gap between \uMQMS{11/5} and \stdsort is only roughly 10\%, and the simulated worst-case of \uMQMS{11/5} is again only slightly over 10\% worse than its average case. 

Notice that while all algorithms have a worst-case guarantee of $\Oh(n\log n)$, the good average case behavior of \stdsort comes with the cost of a bad worst case (see \prettyref{fig:hybrid})~-- the same applies to in-situ Mergesort. Also notice that even the simulated worst case of MoMQuickMergesort is better than the running times of in-situ Mergesort, Wikisort and Heapsort, i.\,e.\ all the other non-hybrid in-place algorithms we tested. 

In \prettyref{fig:Record} and \prettyref{fig:PointerRecord} we measure running times when sorting large objects: in \prettyref{fig:Record} we sort 44-byte records which are compared according to their first 4 bytes. \prettyref{fig:PointerRecord} shows the results when comparing pointers so such records which are allocated on the heap. In both cases \stdsort is the fastest, but MoMQuickMergesort with undersampling is still faster than \stdpartialsort, in-situ Mergesort and Wikisort (unfortunately the latter did not run for sorting pointers). 

In all the experiments, the standard deviations of most algorithms was negligible. Only \stdsort, hybrid QuickMergesort and the simulated worst cases showed a standard deviation which would be visible in the plots. For the worst cases the large standard deviation is because we use \stdnth for choosing bad pivots~-- thus, it is only an artifact of the worst case simulation. The standard deviations of \stdsort and hybrid QuickMergesort can be seen in \prettyref{fig:hybrid} below.
\begin{figure}[tb]%
	\vspace{-5mm}
	\begin{center}
		\scalebox{0.82}{\hspace{-2mm}\input{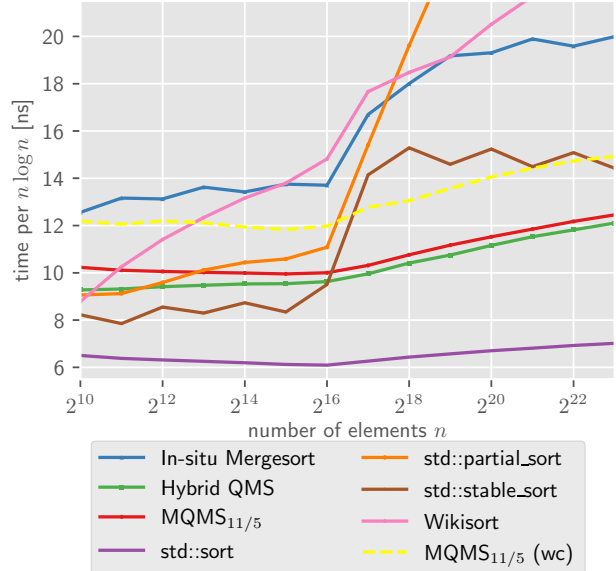}}
	\end{center}
	\vspace{12mm}	
	\caption{\small Running times of MoMQuickMergesort (average and simulated worst case), hybrid QMS and other algorithms for random permutations 44-byte records with 4-byte keys. Running times are divided by $n\log n$.}\label{fig:Record}
\end{figure}%

\begin{figure}[tb]%
	\vspace{-5mm}
	\begin{center}
		\scalebox{0.82}{\hspace{-2mm}\input{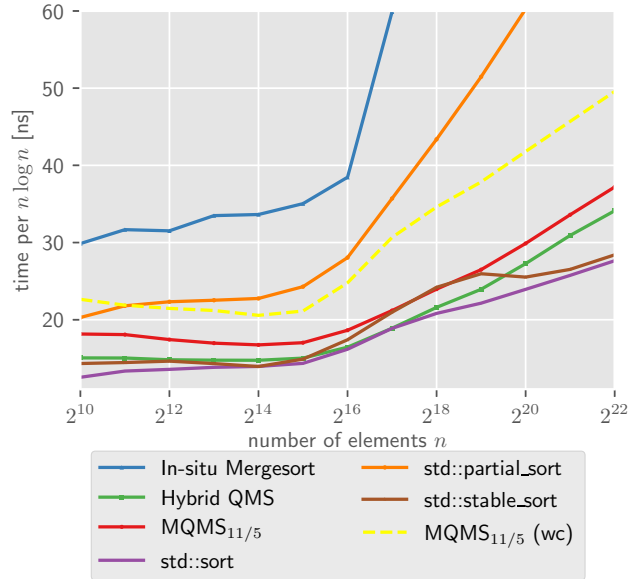}}
	\end{center}
	\vspace{12mm}	
	\caption{\small Running times of MoMQuickMergesort (average and simulated worst case), hybrid QMS and other algorithms for random permutations of pointers to records. Running times are divided by $n\log n$. Wikisort did not work with pointers.}\label{fig:PointerRecord}
\end{figure}%

\begin{figure*}[bht]
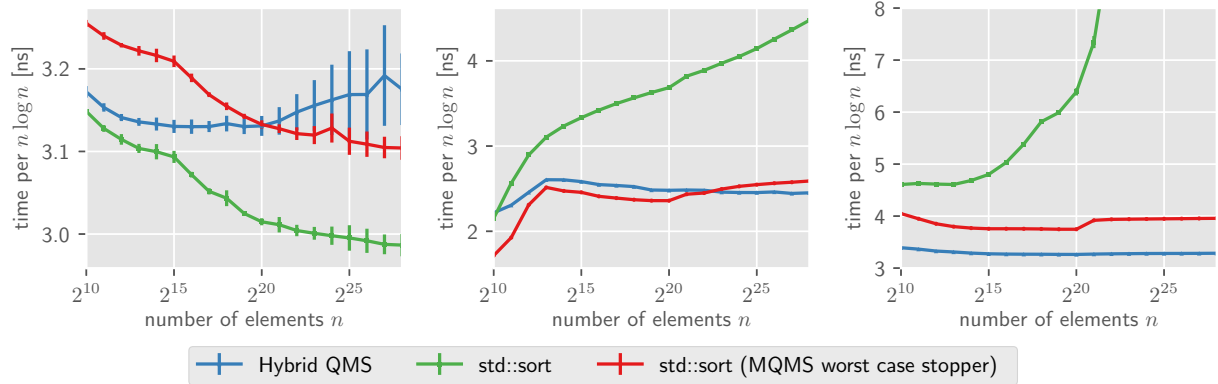
%	
	\vspace{-5mm}
	\scalebox{0.82}{\hspace{4.5mm}\input{hybridAlg.random.pgf}\input{hybridAlg.merge.pgf}\input{hybridAlg.mo3killer.pgf}\hspace{-4.5mm}}
	\vspace{8mm}
	\caption{\small Running times of Introsort (\stdsort) with Heapsort (original) and MoMQuickMergesort as worst-case stopper and hybrid QuickMergesort. Left: random permutation, middle: merge, right: median-of-three killer sequence. The vertical bars represent the standard deviations.}\label{fig:hybrid}
\end{figure*}%

\smallskip
\paragraph{MQMS as worst-case stopper.}
In \prettyref{fig:hybrid} we compare the hybrid algorithms described in \prettyref{sec:hybrid}: hybrid QuickMergesort and Introsort (\stdsort) with MoMQuickMergesort as worst-case stopper. We also include the original \stdsort. We compare three different kinds of inputs: random permutations, merging two sorted runs of almost equal length (the first run two elements longer than the second run), and a median-of-three killer sequence, where the largest elements are located in the middle and back of the array and in between the array is randomly shuffled. 

While for random inputs the difference between the two variants of \stdsort is negligible and also hybrid QuickMergesort is only slightly slower (be aware of the scale of the plot), on the merge permutation we see a considerable difference for large $n$. Except for very small $n$, hybrid QuickMergesort is the fastest here. On the median-of-three killer input, the original \stdsort is outperformed by the other algorithms by a far margin. Also hybrid QuickMergesort is faster than both variants of \stdsort. 

\prettyref{fig:hybrid} also displays the standard deviation as error bars. We see that for the special permutations the standard deviations are negligible~-- which is no surprise. For random permutation hybrid QuickMergesort shows a significantly higher standard deviation than the \stdsort variants. This could be improved by using large pivot samples (\eg\ pseudomedian of 9 or 25). Notice that only for $n\geq 2^{25}$ the standard deviations are meaningful since for smaller $n$ each measurement is already an average so the calculated standard deviation is much smaller than the real standard deviation.

\section{Conclusion}

We have shown that by using the median-of-medians algorithm for pivot selection QuickMergesort turns into a highly efficient algorithm in the worst case, while remaining  competitive on average and fast in practice. 
Future research might address the following points:
\begin{itemize}
	\item 
Although pseudomedians of fifteen elements for sampling the pivot sample seems to be a good choice, other methods could be investigated (\eg\ median of nine).
\item 
The running time could be further improved by using Insertionsort to sort small subarrays.
\item 
Since the main work is done by Mergesort, any tuning to the merging procedure also would directly affect MoMQuickMergesort.
\item 
Also other methods for in-place Mergesort implementations are promising and should be developed further~-- in particular, the (unstable) merging procedure by Chen \cite{Chen06} seems to be a good starting point.
\item To get out the most performance of modern multi-core processors, a parallel version of the algorithm is desirable. For both Quicksort and Mergesort efficient parallel implementations are known. Thus, an efficient parallel implementation of MoMQuickMergesort is not out of reach. However, there is one additional difficulty to overcome: while in the Quicksort recursion both parts can be sorted independently in parallel, in QuickMergesort this is not possible since one part is necessary as temporary memory for sorting the other part with Mergesort.

\end{itemize}

\end{document}